\documentclass[11pt]{article}
\usepackage{sigsam, amsmath}
\usepackage{amsthm}
\usepackage{amssymb}
\usepackage{amsfonts}
\usepackage{mathtools}
\usepackage{maple}
\usepackage[utf8]{inputenc}
\usepackage{breqn}
\usepackage{listings}
\usepackage{moreverb}
\usepackage{verbatim}
\usepackage{sverb}
\usepackage{hyperref}
\usepackage{xcolor}
\newtheorem{theorem}{Theorem}

\newtheorem{example}[theorem]{Example}
\providecommand{\keywords}[1]{\textbf{\textit{Keywords}} #1}

\issue{Vol.xx, No.xx, Issue xxx, month 202x}
\articlehead{ISSAC'22}
\titlehead{FPS In Action}
\authorhead{B. Teguia Tabuguia and W. Koepf}
\begin{document}
	
	\title{{\tt FPS} In Action: An Easy Way To Find Explicit Formulas For Interlaced Hypergeometric Sequences}
	
	\author{Bertrand Teguia Tabuguia\\
		Nonlinear Algebra Group\\
		Max Planck Institute for Mathematics in the Sciences\\
		04103 Leipzig, Germany\\
		{\tt bertrand.teguia@mis.mpg.de}
		
		\and
		
		Wolfram Koepf\\
		Mathematics and Natural Sciences\\
		University of Kassel\\
		Kassel, Germany, 34132\\
		{\tt koepf@mathematik.uni-kassel.de}}
	
	\date{}
	
	\maketitle
	
	\begin{abstract}
		Linear recurrence equations with constant coefficients define the power series coefficients of rational functions. However, one usually prefers to have an explicit formula for the sequence of coefficients, provided that such a formula is ``simple'' enough. Simplicity is related to the compactness of the formula due to the presence of algebraic numbers: ``the smaller, the simpler''. This poster showcases the capacity of recent updates on the Formal Power Series (FPS) algorithm, implemented in Maxima and Maple (\texttt{convert/FormalPowerSeries}), to find simple formulas for sequences like those from \url{https://oeis.org/A307717}, \url{https://oeis.org/A226782}, or \url{https://oeis.org/A226784} by computing power series representations of their correctly guessed generating functions. We designed the algorithm for the more general context of univariate $P$-recursive sequences. Our implementations are available at \url{http://www.mathematik.uni-kassel.de/~bteguia/FPS_webpage/FPS.htm}.
	\end{abstract}
	
	
	\maketitle
	
	\section{Introduction}
	
	Let $\mathbb{K}$ be a field of characteristic zero, usually a finite extension field of the rationals. We call power series representation for a function $f(x)=\sum_{n=0}^{\infty}a_nx^n, a_n\in\mathbb{K}$, a formula for the $(n+1)^{\text{st}}$ summand $a_n x^n$ that one can use to compute any truncation of the power series of $f$. The latter coincide with the Taylor expansion of $f$ when $f$ is analytic at the origin, and thus the representation could also be defined at any point $x_0\in\mathbb{K}$ where $f$ is analytic. In the univariate case, a function $f$ is called holonomic or $D$-finite, if it satisfies a linear differential equation with polynomial coefficients (holonomic DE). From \cite{Koepf1992}, we have a general strategy to search for power series representations symbolically. Given a holonomic function $f(x)=\sum_{n=0}^{\infty}a_nx^n$, as an expression in the variable $x$,
	\begin{enumerate}
		\item compute a holonomic DE satisfied by $f(x)$;
		\item convert the holonomic DE into a linear recurrence equation with polynomial coefficients (holonomic RE);
		\item \textit{solve} the holonomic RE for the coefficients $a_n$.
	\end{enumerate}
	The two last steps are equivalent to finding power series solutions of linear ordinary differential equations. One cannot always solve the resulting recurrence equation. 
	
	For $m\in\mathbb{N}$, a term $a_n$ is called $m$-fold hypergeometric if the ratio $a_{n+m}/a_n\in\mathbb{K}(n)$. When $m$ is not specified, $m$-fold hypergeometric denotes all such terms for arbitrary positive integers $m$. We say that a function is of hypergeometric type if its power series coefficients are evaluations of $m$-fold hypergeometric terms. The definition extends to Laurent-Puiseux series, but in this paper we only consider power series with non-negative integer exponents. Numerous holonomic functions, including rational functions, are of hypergeometric type. The computation of $m$-fold hypergeometric term solutions of holonomic REs is effective (see \cite{ryabenko2002formal,BTphd,teguia2022symbolic}). On this poster, we use the algorithm \texttt{mfoldHyper} (available in Maple 2022 as \texttt{LREtools-mhypergeomsols}), developed in \cite{teguia2022symbolic} for the efficient computation of hypergeometric type power series (see also \cite{teguia2020power,tabuguia2021hypergeometric}). It is worth mentioning that \texttt{mfoldHyper} extends the algorithms by Petkov\v{s}ek and Mark van Hoeij (see \cite{petkovvsek1992hypergeometric,van1999finite,cluzeau2006computing}) and has a much better performance than some previous approaches in the same direction (see \cite{petkovvsek1993finding,horn2012m}). Our Formal Power Series (FPS) algorithm uses \texttt{mfoldHyper} to compute a basis for all $m$-fold hypergeometric term solutions on the third step given above, and uses linear algebra with a truncated series expansion of $f$ to deduce a power series representation for $f$. For more details about the algorithm, we refer the reader to \cite{teguia2022symbolic,BTphd}.
	\begin{example} The FPS algorithm computes the following power series representation for $\frac{1}{\left((x^2-p)(x^3-q)\right)}$ for arbitrary constants $p$ and $q$.
		\begin{dmath}\label{eq1}
			\moverset{\infty}{\munderset{n =0}{\textcolor{gray}{\sum}}}\! \left(-\frac{\left(q \,p^{-1-\frac{n}{2}}-q^{-\frac{1}{3}-\frac{n}{3}} p \right) x^{n}}{p^{3}-q^{2}}\right)+\moverset{\infty}{\munderset{n =0}{\textcolor{gray}{\sum}}}\! \left(-\frac{\left(p^{\frac{3}{2}}-q \right) p^{-n -\frac{3}{2}} x^{2 n +1}}{p^{3}-q^{2}}\right)+\moverset{\infty}{\munderset{n =0}{\textcolor{gray}{\sum}}}\! \left(-\frac{q^{-1-n} p \left(q^{\frac{2}{3}}-p \right) x^{3 n}}{p^{3}-q^{2}}\right)+\left(\moverset{\infty}{\munderset{n =0}{\textcolor{gray}{\sum}}}\! \frac{\left(q^{\frac{2}{3}}-p \right) q^{-n -\frac{2}{3}} x^{3 n +1}}{p^{3}-q^{2}}\right).
		\end{dmath}
		For software reason, the formula is not identical with the one obtained with Maple 2022, although correct; however, the output is still much more compact compared to previous Maple versions. Our implementation is available for Maple and Maxima users at \url{http://www.mathematik.uni-kassel.de/~bteguia/FPS_webpage/FPS.htm}. The authors welcome any comments for the improvement of the package.
	\end{example}
	
	The latter example presents a typical situation of what happens when it comes to computing explicit formulas for power series coefficients of rational functions that generate some sequences from N. J. A. Sloane \url{https://oeis.org}. FPS splits the formula modulo some integers, which allows to also deal with the situation of many zeros (see \cite{kauers2019you}) in the sequence with no specific care about them. The Padé approximation is often the best choice to guess the rational function that generates a sequence \cite{geddes1979symbolic}. In our next examples, we use the implementation in the \texttt{Gfun} package (\texttt{ratpoly}) (\cite{gfun}) to find the generating functions, and FPS to compute the desired formulas.
	
	\section{Some Explicit Formulas}
	
	In what follows, we compute power series representations of rational functions generating some sequences from \url{https://oeis.org}. The resulting representations give explicit formulas for the $(n+1)^{\text{st}}$ terms (we start sequences at $0$) of the corresponding sequences. We do not give detailed proofs that the guessed rational function is the correct one; however, the formula obtained using FPS is correct by the correctness of the FPS algorithm. Throughout this section $a_n$ will denote the $(n+1)^{\text{st}}$ term of the sequence for each example.
	
	\begin{example} Let us consider the sequence \textbf{A307717} from \url{https://oeis.org/A307717}. $a_n$ counts the number of palindromic squares, $k^2$, of length $n+1$ such that $k$ is also palindromic.
		
		\begin{Maple Normal}
			We use the 33 first terms of the sequence. That is the minimal required by the guess in the next step.
		\end{Maple Normal}
	
\begin{lstlisting}
 > L:=[4, 0, 2, 0, 5, 0, 3, 0, 8, 0, 5, 0, 13, 0, 9, 0, 22, 0, 
 16, 0, 37, 0, 27, 0, 60, 0, 43, 0, 93, 0, 65, 0, 138]:
\end{lstlisting}
	
		\begin{Maple Normal}
			We guess the generating function.
		\end{Maple Normal}
	
		\begin{lstlisting}
			> f:=gfun:-listtoratpoly(L,x)[1]
		\end{lstlisting}
	
		\begin{dmath}\label{(2)}
			f \coloneqq -\frac{2 x^{16}-x^{14}-5 x^{12}+5 x^{10}+12 x^{8}-5 x^{6}-11 x^{4}+2 x^{2}+4}{-x^{16}+4 x^{12}-6 x^{8}+4 x^{4}-1}
		\end{dmath}
	
		\begin{Maple Normal}
			We compute the power series representation using FPS. The input is the expression $f$, its variable $x$, and a summation variable $n$ chosen by the user.
		\end{Maple Normal}
	
		\begin{lstlisting}
			> FPS(f,x,n)
		\end{lstlisting}
	
		\begin{dmath}\label{eq2}
			2+\left(\moverset{\infty}{\munderset{n =0}{\textcolor{gray}{\sum}}}\! \left(\frac{\left(-1\right)^{n} n^{3}}{96}-\frac{3 \left(-1\right)^{n} n^{2}}{32}+\frac{65 \left(-1\right)^{n} n}{96}+\frac{7 \left(-1\right)^{n}}{32}+\frac{n^{3}}{32}-\frac{5 n^{2}}{32}+\frac{37 n}{32}+\frac{57}{32}\right) x^{2 n}\right)
		\end{dmath}
	\end{example}
	
	\begin{theorem} The sequence \textbf{A307717} from \url{https://oeis.org/A307717} has the explicit formula
		\begin{align}\label{eq3}
			&a_0 = 4 \nonumber\\
			&a_{2n+1} = 0, n\geq 0\\
			&a_{2n} = \frac{\left(\left(-1\right)^{n} +3\right) n^{3}-\left(9 \left(-1\right)^{n} + 15\right) n^{2}+\left(65 \left(-1\right)^{n} +111 \right) n +21 \left(-1\right)^{n} +171}{96}, n\geq 1, 
		\end{align}
		where $a_n$ is its $(n+1)^{\text{st}}$ term.
	\end{theorem}
	\begin{proof}It is enough to prove that the generating function $(\ref{(2)})$ is the correct one. From the connection to the sequence from \url{https://oeis.org/A218035} whose generating function, denote it by $g(x)$, is known, one verifies that $f$ in $(\ref{(2)})$ and $g$ are linked by the relation $f(x)=g(x^2)/x^2$, which holds. 
	\end{proof}

	We noticed that the latter example is also investigated in the recent paper \cite{kauers2022guessing} about guessing. There the authors also provide a different formula for \textbf{A307717}.
	
	\begin{example} Our second sequence is \textbf{A226782} from \url{https://oeis.org/A226782}. $a_n=0$ if $n+1$ is even, and the inverse of $4$ in the ring $\mathbb{Z}/(n+1)\mathbb{Z}^*$ if $n+1$ is odd. We proceed as before and find an explicit formula for $a_n$.
		
\begin{lstlisting}
	> L:=[0, 0, 1, 0, 4, 0, 2, 0, 7, 0, 3, 0, 10, 0, 4, 0, 13]:
\end{lstlisting}
\begin{lstlisting}
	> f:=gfun:-listtoratpoly(L,x)[1]
\end{lstlisting}
\begin{dmath}\label{(8)}
	f \coloneqq -\frac{-x^{8}+4 x^{4}+x^{2}}{-x^{8}+2 x^{4}-1}
\end{dmath}
\begin{lstlisting}
	> FPS(f,x,n)
\end{lstlisting}
\begin{dmath}\label{(9)}
	-1+\left(\moverset{\infty}{\munderset{n =0}{\textcolor{gray}{\sum}}}\! \left(\frac{\left(-1\right)^{n} n}{2}+\frac{\left(-1\right)^{n}}{4}+n +\frac{3}{4}\right) x^{2 n}\right)
\end{dmath}
In this example, the representation can be rewritten without the extra $-1$ since the series part evaluates to $1$ at $n=0$.
	\end{example}
	
	\begin{theorem} The sequence \textbf{A226782} from \url{https://oeis.org/A226782} has the explicit formula 
		\begin{align}
			&a_0 = 0 \nonumber\\
			&a_{2n+1} = 0, n\geq 0 \\
			&a_{2n} = \frac{\left(-1\right)^{n} n}{2}+\frac{\left(-1\right)^{n}}{4}+n +\frac{3}{4}, n\geq 1, \label{eq4}
		\end{align}
		where $a_n$ is its $(n+1)^{\text{st}}$ term.
	\end{theorem}
	\begin{proof}
		One verifies that the generating function given in \url{https://oeis.org/A226782} is a shift of $f$ in $(\ref{(8)})$, as we start the sequence at $0$ instead of $1$. 
	\end{proof}
	
	\section{Conclusion}
	This poster aims to recommend our FPS implementation to the community of computer algebraists or scientists interested in finding explicit formulas for sequences. Combining it with guessing strategies enables one to discover or recover explicit formulas. Of course, limitations may always exist, at least in practice, when implementing the theory; however, as research goes on, FPS also keeps improving to produce explicit formulas for larger classes of functions in the future. We mention that previous versions of the \texttt{convert/FormalPowerSeries} package could not compute the representations presented here. However, this is now possible with Maple 2022, which incorporates our implementation thanks to Jürgen Gerhard from Maplesoft.

\end{document}